\expandafter\ifx\csname ifarxiv\endcsname\relax
  \expandafter\newif\csname ifarxiv\endcsname
  \csname arxivtrue\endcsname
\fi
\documentclass[lettersize,journal]{IEEEtran}
\usepackage{amsmath,amsfonts,amsthm,amssymb,mathtools}
\usepackage[ruled,vlined,linesnumbered,noend]{algorithm2e}
\usepackage{array}
\usepackage[caption=false,font=normalsize,labelfont=sf,textfont=sf]{subfig}
\usepackage{textcomp}
\usepackage{stfloats}
\usepackage{url}
\usepackage{verbatim}
\usepackage{graphicx}
\usepackage{cite}
\newenvironment{constraintblock}{%
  \setlength{\abovedisplayskip}{3pt plus 1pt minus 1pt}%
  \setlength{\belowdisplayskip}{3pt plus 1pt minus 1pt}%
  \setlength{\abovedisplayshortskip}{3pt plus 1pt minus 1pt}%
  \setlength{\belowdisplayshortskip}{3pt plus 1pt minus 1pt}%
}{\ignorespacesafterend}
\newenvironment{conditions}{%
  \renewcommand{\theenumi}{(\roman{enumi})}%
  \begin{enumerate}[\IEEEsetlabelwidth{(iii)}]%
}{\end{enumerate}}

\DontPrintSemicolon
\SetKwFunction{CP}{Algorithm \ref{alg:cp-lodge}}

\newtheorem{theorem}{Theorem}
\newtheorem{proposition}{Proposition}
\newtheorem{lemma}{Lemma}
\newtheorem{definition}{Definition}

\newtheorem{assumption}{Assumption}

\begin{document}
\bstctlcite{IEEEexample:BSTcontrol}

\title{A Deficiency-Based Approach for Unbalanced Three-Phase Distribution System Expansion Planning}

\author{Eli~Brock,~\IEEEmembership{Student Member,~IEEE},
        Miguel~Heleno,~\IEEEmembership{Senior Member,~IEEE},
        Javad~Lavaei,~\IEEEmembership{Fellow,~IEEE},
        and Somayeh~Sojoudi,~\IEEEmembership{Fellow,~IEEE}%
\thanks{\ifarxiv\else Manuscript received [date]; revised [date]. \fi This work was supported in part by the U.S. Army Research Laboratory and the U.S. Army Research Office under Grant W911NF2010219, the Office of Naval Research, and NSF.}%
\thanks{E. Brock and S. Sojoudi are with the Department of Electrical Engineering and Computer Sciences, University of California, Berkeley, CA 94720 USA (e-mail: eli.brock@berkeley.edu; sojoudi@berkeley.edu).}%
\thanks{J. Lavaei is with the Department of Industrial Engineering and Operations Research, University of California, Berkeley, CA 94720 USA (e-mail: lavaei@berkeley.edu).}%
\thanks{M. Heleno is with Lawrence Berkeley National Laboratory, Berkeley, CA 94720 USA (e-mail: MiguelHeleno@lbl.gov).}%
}

\ifarxiv
  \def\paperjournal{}
\else
  \def\paperjournal{IEEE Transactions on Power Systems}
\fi
\markboth{\paperjournal}%
{Brock \MakeLowercase{\textit{et al.}}: Deficiency-Based Unbalanced Distribution Expansion Planning}

\maketitle

\begin{abstract}
Distribution planners typically characterize grid needs through power flow studies and fit upgrades to those needs by hand. More principled one-shot optimization methods either do not scale or misrepresent unbalanced three-phase networks. This paper proposes an iterative planning process that alternates between two problem classes with mature off-the-shelf solvers. An optimal power flow (OPF) problem quantifies the deficiency of a candidate portfolio by optimizing the dispatch of controllable devices and measuring thermal and voltage violations on a common scale, with an optimal value of zero exactly at adequate portfolios. The OPF multipliers give a first-order approximation of the relief each investment candidate would provide. We bound the error of this approximation and show that structural properties of the deficiency guard it against overestimation. Building on this property, a heuristic cutting-plane search uses a mixed-integer linear program to propose each next portfolio from these approximations. On six unbalanced three-phase feeders from the SMART-DS dataset, the search outperforms a linearized optimization-based benchmark, and the results agree across five variants of the deficiency.
\end{abstract}


\section{Introduction} \label{sec:intro}
Traditionally, electric distribution networks served as unidirectional, passive power-delivery systems, and utilities added, replaced, or upgraded their infrastructure as load grew. The ongoing emergence of electric vehicles, HVAC electrification, distributed photovoltaic (PV) generation, and load growth has sparked a costly buildout of distribution infrastructure. The US Energy Information Administration calls distribution spending ``the main driver of electricity spending increases over the last two decades'', far outpacing capital expenditures in generation and transmission \cite{eia2024grid}. The scale of these upgrades calls for principled planning methods to contain costs.

In practice, distribution planning proceeds as an engineering loop: planners run power flow studies of the network under a forecast peak load, locate the components whose operating limits are violated, and develop upgrades that mitigate each violation \cite{homer2020pnnl,pilo2012cigre}. California's investor-owned utilities, for example, publish an annual grid needs assessment listing each forecast deficiency by facility, quantified as an overload relative to the facility's rating or an exceedance of the voltage limits, to identify opportunities for non-wires alternatives to defer traditional upgrades \cite{pge2023gna}. Selecting the upgrades remains largely a matter of engineering judgment.

Optimization-based distribution system expansion planning (DSEP) \cite{rastgou2024,georgilakis2015review} instead couples long-term, typically binary, investment decisions with an operational power flow subproblem in a single mathematical program. Unlike transmission systems, which are approximately balanced and well approximated by linearized power flow models, distribution systems are highly unbalanced, often dominated by single-phase laterals, and exhibit non-negligible voltage drops and resistive losses. Existing formulations nonetheless adopt the second-order cone relaxation \cite{amjady2018} or the LinDistFlow approximation \cite{moreira2024} of the branch flow model, or other linearized network models \cite{tabares2016,munoz2016,munoz2018}, and all apply only to balanced single-phase equivalent networks. Section \ref{sec:formulation} introduces the DSEP problem and discusses the limitations of relaxations and linearizations for planning more broadly.

Several lines of work automate parts of the planning loop, for example by running power flows, iteratively upgrading each overloaded line or transformer to the next standard size, and adjusting voltage control settings to reduce the number of voltage violations \cite{horowitz2018cost,palmintier2021la100}. Metaheuristic planners for unbalanced secondary circuits score each candidate plan with a three-phase power flow and penalize constraint violations \cite{cossi2005,cossi2009,mendoza2013}. Heleno et al. \cite{heleno2024} alternate between a three-phase power flow that identifies upgrade candidates near violated limits and a balanced least-cost expansion model with slack variables. In each case, inadequacy is measured by a proxy suited to the problem at hand: a count of violations, a hand-weighted penalty, or the slacks of an approximate model. None offers, under the full network physics, a consistent measure of how far a plan is from adequate or an indication of which investments would bring it closer.

Motivated by contingency analysis and planning, Overbye \cite{overbye1994} proposed minimizing the residual of the power flow equations to measure how far a given load profile lies from solvability for a given network. Jereminov et al. \cite{jereminov2020} reinterpreted the residual as fictitious, penalized current sources at each node. Foster et al. \cite{foster2022} extended the idea to unbalanced three-phase distribution networks. The aforementioned authors suggest that the optimal solutions of these feasibility problems may help planners manually locate and compensate for deficiency in overloaded networks. Zamzam et al. \cite{zamzam2018} attached slacks to the voltage limits as well as the power balance of a multi-phase optimal power flow (OPF), so that the slacks point to the constraints that render it infeasible. Following this usage, we call the optimal value of such a slack-minimizing problem a \textit{deficiency}.

In summary, the existing literature does not yet offer a satisfactory approach to distribution system planning. One-shot optimization-based methods either do not scale or do not capture the unbalanced three-phase physics, while power-flow-based planning loops lack principled rules for choosing appropriate upgrades. Deficiency-based analysis is promising, but has yet to be integrated into planning. This work proposes an iterative planning process that mirrors existing utility practice, alternating between two problems with mature off-the-shelf solvers: an OPF instance that takes the place of the power flow study, quantifying the deficiency of a given portfolio and locating the violations; and a mixed-integer linear program that uses the sensitivity information encoded in the optimal multipliers of the OPF to propose the next portfolio. Specifically, our contributions are:
\begin{enumerate}
    \item We generalize and formalize the notion of deficiency (Section \ref{sec:theory}), arriving at a definition that allows for different sources of slack and objective forms and can be evaluated with off-the-shelf interior-point OPF solvers.
    \item We provide an interpretable upper bound on the error incurred by using the first-order approximation of a deficiency function to predict the deficiency at portfolios that have not been evaluated (Section \ref{sec:validity}).
    \item Motivated by this bound, we design a deficiency-guided cutting-plane portfolio search (Section \ref{sec:cutting}). In experiments on six unbalanced feeders from the SMART-DS dataset (Section \ref{sec:experiments}), the search outperforms a linearized one-shot benchmark, with consistent results across five deficiency variants.
\end{enumerate}

\textit{Notation:} $\mathbb{R}$ and $\mathbb{C}$ are the real and complex numbers, and $\mathbb{R}_+$ and $\mathbb{R}_{++}$ are the nonnegative and positive reals. $\mathbf{j}$ is the imaginary unit, $\operatorname{Re}$ and $\operatorname{Im}$ take real and imaginary parts, and $\bullet^\top$ and $\bullet^H$ are the transpose and conjugate transpose. $\lvert \cdot \rvert$ is the modulus, applied elementwise to vectors, or the cardinality of a set; $\lVert \cdot \rVert_p$ is the $p$-norm, with $\lVert \cdot \rVert \equiv \lVert \cdot \rVert_2$. $\mathbf{1}$ and $\mathbf{0}$ are the all-ones and all-zeros vectors of appropriate dimension, $\text{Diag}(x)$ is the diagonal matrix with diagonal $x$, and $\operatorname{proj}_{x^\perp} y \equiv y - x \, (x^\top y) / (x^\top x)$ is the component of $y$ orthogonal to $x$. $\operatorname{conv}$ denotes the convex hull, $f(t^+)$ and $f(t^-)$ denote the right and left limits of $f$ at $t$, $f'(t^\pm)$ denote the one-sided derivatives, and $\equiv$ denotes a definition.

\section{Distribution System Expansion Planning} \label{sec:formulation}
Consider an electric distribution system with a substation bus $0$ and a set of non-substation buses $\mathcal{N}$; let $\mathcal{N}_+ \equiv \mathcal{N} \cup \{0\}$ denote the set of all buses. Each bus $i \in \mathcal{N}_+$ has a set of phases $\Phi_i \subseteq \{a, b, c\}$ and hosts a (possibly empty) set of loads $\mathcal{D}_i$, where each load $m \in \mathcal{D}_i$ serves a subset of phases $\Phi_m \subseteq \Phi_i$; write $\mathcal{D}_{i,\phi} \equiv \left\{m \in \mathcal{D}_i : \phi \in \Phi_m\right\}$ for those serving phase $\phi$ and $\mathcal{D} \equiv \bigcup_{i \in \mathcal{N}} \mathcal{D}_i$ for the set of all loads. Each non-substation bus $i \in \mathcal{N}$ also hosts a (possibly empty) set of capacitor banks $\mathcal{C}_i$, each serving every phase of its bus; write $\mathcal{C} \equiv \bigcup_{i \in \mathcal{N}} \mathcal{C}_i$ for the set of all capacitor banks. Buses are connected by three pairwise disjoint sets of branches: lines $\mathcal{L}$, transformers $\mathcal{T}$, and voltage regulators $\mathcal{R}$, each a subset of $\mathcal{N}_+ \times \mathcal{N}_+$. Let $\mathcal{E} \equiv \mathcal{L} \cup \mathcal{T} \cup \mathcal{R}$ denote the set of branches (edges). By convention, edges point away from bus $0$ along a spanning tree rooted there. We assume that, under this convention, the phases of downstream buses are a subset of those of upstream buses: for all $i \to j \in \mathcal{E}$, $\Phi_j \subseteq \Phi_i$.

In general, every utility-owned device on the network is chosen from a menu of mutually exclusive installation candidates: define the set of \textit{investment targets} $\mathcal{I} \equiv \mathcal{E} \cup \mathcal{C}$ and, for each target $g \in \mathcal{I}$, the menu $\mathcal{G}_g$ of candidates from which exactly one is installed.
Formally,
\begin{subequations} \label{eq:binary_constraints}
\begin{constraintblock}
    \begin{flalign}
        \label{eq:integrality} & b_k \in \left\{0, 1\right\} && \forall g \in \mathcal{I}, k \in \mathcal{G}_g
    \end{flalign}
    \begin{flalign}
        \label{eq:single_upgrade} & \sum_{k \in \mathcal{G}_g} b_k = 1 && \forall g \in \mathcal{I}
    \end{flalign}
    \begin{flalign}
        \label{eq:side} & \Gamma b \leq \gamma &&
    \end{flalign}
\end{constraintblock}
\end{subequations}
where $b_k$ indicates that candidate $k \in \mathcal{G}_g$ is installed, and \eqref{eq:side} collects any further linear constraints on the portfolio, such as siting rules, through a matrix $\Gamma$ and a vector $\gamma$ of compatible dimensions. Let $\mathcal{B} \equiv \{b : \eqref{eq:binary_constraints}\}$ denote the set of feasible portfolios, and define the investment cost $C(b) \equiv \sum_{g \in \mathcal{I}} \sum_{k \in \mathcal{G}_g} c_k b_k$, where $c_k \in \mathbb{R}_+$ is the cost of candidate $k$.

The binary variables enter the network model through the device parameters, selected from the menus as
\begin{constraintblock}
    \begin{flalign*}
        & \left(Y_{ij}, \overline{I}_{ij}\right) \equiv \sum_{k \in \mathcal{G}_{ij}} b_k \left(Y_k, \overline{I}_k\right) && \forall i \to j \in \mathcal{L}
    \end{flalign*}
    \begin{flalign*}
        & \left(Y_{ij}, \overline{S}_{ij}\right) \equiv \sum_{k \in \mathcal{G}_{ij}} b_k \left(Y_k, \overline{S}_k\right) && \forall i \to j \in \mathcal{T}
    \end{flalign*}
    \begin{flalign*}
        & \left(Z_{ij}, \tau_{ij}\right) \equiv \sum_{k \in \mathcal{G}_{ij}} b_k \left(Z_k, \tau_k\right) && \forall i \to j \in \mathcal{R}
    \end{flalign*}
    \begin{flalign*}
        & \overline{B}_{m} \equiv \sum_{k \in \mathcal{G}_{m}} b_k \overline{B}_k && \forall m \in \mathcal{C}
    \end{flalign*}
\end{constraintblock}
$Y_{ij}$ is the series admittance of branch $i \to j$ ($\mathbb{C}^{\left|\Phi_j\right| \times \left|\Phi_j\right|}$ for a line, $\mathbb{C}$ for a transformer), $\overline{I}_{ij} \in \mathbb{R}_+$ is the ampacity of line $i \to j$, $\overline{S}_{ij} \in \mathbb{R}_+$ is the per-phase apparent power rating of transformer $i \to j$, $Z_{ij} \in \mathbb{C}$ is the series impedance of regulator $i \to j$, $\tau_{ij} \in \mathbb{R}_+$ is the maximum tap deviation of regulator $i \to j$, and $\overline{B}_m \in \mathbb{R}_+$ is the maximum susceptance of capacitor $m$. The same symbols with a candidate subscript $k \in \mathcal{G}_g$, $g \in \mathcal{I}$, are the corresponding candidate parameters. An uninstalled conductor or transformer candidate is an open circuit ($Y_k \to 0$), while an uninstalled regulator candidate is a short circuit ($Z_k \to 0$).

Every device parameter in the network model belongs to some candidate, so an existing device is modeled as a zero-cost candidate. Conductor and transformer menus typically include the existing device, while the menu of a regulator or capacitor without an existing device, a prospective installation site, includes a ``do-not-invest'' proxy candidate with $c_k = 0$ and $\tau_k = \varepsilon$, $Z_k = 0$ or $\overline{B}_k = \varepsilon$, respectively, for a small $\varepsilon > 0$ that keeps the device's dispatch bounds separated when it is not installed (required to satisfy the constraint qualification in Section \ref{sec:theory}). A target that is not a true potential investment has a singleton menu holding the existing device; a bus offering neither an existing bank nor an installation site has $\mathcal{C}_i = \emptyset$.

The following constraints encode the power flow physics:
\begin{subequations} \label{eq:power_flow}
\begin{constraintblock}
    \begin{flalign}
        \nonumber & \sum_{m \in \mathcal{D}_{i,\phi}} \tilde{I}_{m, \phi} - \sum_{m \in \mathcal{D}_{i,\phi} \cup \mathcal{C}_i} I_{m,\phi} \\
        \label{eq:kcl} & \quad = \sum_{j: j \to i} I_{ji,\phi} + \sum_{j:i \to j} I_{ij, \phi} && \forall i \in \mathcal{N}, \phi \in \Phi_i
    \end{flalign}
    \begin{flalign}
        \label{eq:line_ohms} & I_{ij} = Y_{ij}\left(V_i - V_j\right) && \forall i \to j \in \mathcal{L}
    \end{flalign}
    \begin{flalign}
        \label{eq:line_rev} & I_{ji,\phi} = -I_{ij,\phi} && \forall i \to j \in \mathcal{L}, \phi \in \Phi_j
    \end{flalign}
    \begin{flalign}
        \label{eq:xfmr_ohms} & I_{ij,\phi} = Y_{ij}\left(V_{i,\phi} - n_{ij}V_{j,\phi}\right) && \forall i \to j \in \mathcal{T}, \phi \in \Phi_j
    \end{flalign}
    \begin{flalign}
        \label{eq:xfmr_rev} & I_{ji,\phi} = -n_{ij}I_{ij,\phi} && \forall i \to j \in \mathcal{T}, \phi \in \Phi_j
    \end{flalign}
    \begin{flalign}
        \label{eq:reg_ohms} & V_{i,\phi} - \frac{V_{j,\phi}}{t_{ij,\phi}} = Z_{ij}I_{ij,\phi} && \forall i \to j \in \mathcal{R}, \phi \in \Phi_j
    \end{flalign}
    \begin{flalign}
        \label{eq:reg_rev} & I_{ji,\phi} = -\frac{I_{ij,\phi}}{t_{ij,\phi}} && \forall i \to j \in \mathcal{R}, \phi \in \Phi_j
    \end{flalign}
    \begin{flalign}
        \label{eq:load} & V_{i,\phi}I_{m,\phi}^H = S_{m,\phi}^\dagger && \forall i \in \mathcal{N}, m \in \mathcal{D}_i, \phi \in \Phi_m
    \end{flalign}
    \begin{flalign}
        \label{eq:capacitor_current} & I_{m,\phi} = \mathbf{j}B_mV_{i,\phi} && \forall i \in \mathcal{N}, m \in \mathcal{C}_i, \phi \in \Phi_i
    \end{flalign}
    \begin{flalign}
        \label{eq:slack_voltage} & V_{0,\phi} = V_{0,\phi}^\dagger && \forall \phi \in \Phi_0
    \end{flalign}
\end{constraintblock}
\end{subequations}
$V_{i, \phi} \in \mathbb{C}$ and $V_{i,\phi}^\dagger \in \mathbb{C}$ are the voltage and nominal voltage at bus $i \in \mathcal{N}$ and phase $\phi \in \Phi_i$. $I_{m,\phi} \in \mathbb{C}$, $S_{m,\phi}^\dagger \in \mathbb{C}$ and $\tilde{I}_{m,\phi} \in \mathbb{C}$ are the current demand, net complex power demand (subsuming distributed photovoltaic generation) and \textit{compensation current} injection of load $m$ at phase $\phi \in \Phi_m$. $I_{m,\phi} \in \mathbb{C}$ is the current drawn by capacitor $m \in \mathcal{C}_i$ at phase $\phi \in \Phi_i$. The compensation currents $\tilde{I}$ represent fictitious load shed (or augmentation), used by the deficiency programs in Section \ref{sec:deficiency}. $B_m \in \mathbb{R}$ is the dispatchable susceptance of capacitor $m \in \mathcal{C}_i$, fixed across phases. $I_{ij,\phi} \in \mathbb{C}$ and $I_{ji,\phi} \in \mathbb{C}$ are the currents drawn from the sending and receiving buses of branch $i \to j \in \mathcal{E}$ at phase $\phi \in \Phi_j$. $n_{ij} \in \mathbb{R}_{++}$ is the fixed nominal turns ratio of transformer $i \to j \in \mathcal{T}$, and $t_{ij,\phi} \in \mathbb{R}_{++}$ is the dispatchable tap ratio of regulator $i \to j \in \mathcal{R}$ at phase $\phi$. In \eqref{eq:line_ohms}, $V_i$, $V_j$ and $I_{ij}$ are vectors over the phases $\Phi_j$, ordered $a$, $b$, $c$ as in $Y_k$.
The devices are subject to dispatch bounds,
\begin{subequations} \label{eq:device_bounds}
\begin{constraintblock}
    \begin{flalign}
        \label{eq:capacitor_bound} & 0 \leq B_{m} \leq \overline{B}_{m} && \forall m \in \mathcal{C}
    \end{flalign}
    \begin{flalign}
        \label{eq:tap_range} & \left|t_{ij,\phi} - 1\right| \leq \tau_{ij} && \forall i \to j \in \mathcal{R}, \phi \in \Phi_j
    \end{flalign}
\end{constraintblock}
\end{subequations}
and the network is subject to operating limits on voltage, conductor current and transformer loading. Let $\overline{\mathcal{M}} \equiv \mathcal{N}_+ \cup \mathcal{L} \cup \mathcal{T}$ collect the components subject to an operating limit. The limits are imposed on a subset $\mathcal{M} \subseteq \overline{\mathcal{M}}$ of \textit{monitored} components, with $\mathcal{N}_{\mathcal{M}} \equiv \mathcal{N}_+ \cap \mathcal{M}$, $\mathcal{L}_{\mathcal{M}} \equiv \mathcal{L} \cap \mathcal{M}$ and $\mathcal{T}_{\mathcal{M}} \equiv \mathcal{T} \cap \mathcal{M}$:
\begin{subequations} \label{eq:inequalities}
\begin{constraintblock}
    \begin{flalign}
        \nonumber & \underline{V}_i - \sigma_{i,\phi}\left(\left|V_{i,\phi}^\dagger\right| - \underline{V}_i\right) \leq \left|V_{i,\phi}\right| \\
        \label{eq:volt_lims} & \quad \leq \overline{V}_i + \sigma_{i,\phi}\left(\overline{V}_i - \left|V_{i,\phi}^\dagger\right|\right) && \forall i \in \mathcal{N}_{\mathcal{M}}, \phi \in \Phi_i
    \end{flalign}
    \begin{flalign}
        \label{eq:current_lim} & \left|I_{ij,\phi}\right| \leq \left(1 + \sigma_{ij,\phi}\right)\overline{I}_{ij} && \forall i \to j \in \mathcal{L}_{\mathcal{M}}, \phi \in \Phi_j
    \end{flalign}
    \begin{flalign}
        \label{eq:apparent_power_lim} & \left|I_{ji,\phi}^HV_{j,\phi}\right| \leq \left(1 + \sigma_{ij,\phi}\right)\overline{S}_{ij} && \forall i \to j \in \mathcal{T}_{\mathcal{M}}, \phi \in \Phi_j
    \end{flalign}
\end{constraintblock}
\end{subequations}
$\underline{V}_i \in \mathbb{R}$ and $\overline{V}_i \in \mathbb{R}$ are the lower and upper voltage magnitude bounds at bus $i \in \mathcal{N}_+$, applied across phases. The \textit{exceedances} $\sigma \geq 0$, a second family of slack alongside $\tilde{I}$, are dimensionless and therefore commensurable across constraint types. Planning ordinarily imposes the limits of every component, $\mathcal{M} = \overline{\mathcal{M}}$; we write \eqref{eq:inequalities} for a generic monitored set so that the method of Sections \ref{sec:theory} and \ref{sec:cutting} can isolate the effect of specific limits.

Collecting the operational variables into $x \equiv \left(V, I, B, t, \tilde{I}, \sigma\right)$, we abbreviate the physics \eqref{eq:power_flow} as $h(x, b) = 0$ and the dispatch bounds and operating limits \eqref{eq:device_bounds}--\eqref{eq:inequalities} as $g_{\mathcal{M}}(x, b) \leq 0$. We represent each complex variable and equation by its real and imaginary parts, so that $h$ and $g_{\mathcal{M}}$ are real-valued functions of real variables. The full distribution system expansion planning problem is
\begin{subequations}  \label{eq:lodge}
\begin{align}
    \min_{x, b} \quad & C(b) \\
    \text{s.t.} \quad & h\left(x, b\right) = 0, \quad g_{\overline{\mathcal{M}}}\left(x, b\right) \leq 0 \\
    & \tilde{I} = \mathbf{0}, \quad \sigma = \mathbf{0} \\
    & b \in \mathcal{B}
\end{align}
\end{subequations}
For simplicity, and to match the devices present in the networks of Section \ref{sec:experiments}, \eqref{eq:lodge} is restricted to reconductoring, transformer upgrades, capacitor banks, and voltage regulators, and to a single planning load $S^\dagger$, which may represent the planner's worst-case coincident load. In practice, capacitor banks switch in discrete stages and regulator taps move in discrete steps; \eqref{eq:lodge} approximates both as continuous over their ranges. The framework of Section \ref{sec:theory} and its results naturally extend to other candidates, such as front-of-meter storage or demand response, and to multiple load scenarios. \eqref{eq:lodge} captures the coupling between the voltage and capacity problems. A larger conductor also lowers a line's impedance and, with it, the voltage drop, which may defer voltage regulation. Similarly, reactive injections from capacitor banks may reduce flow through upstream branches, potentially deferring capacity upgrades. See the results in Section \ref{sec:results} for examples of such trade-offs.

\eqref{eq:lodge} is a nonconvex mixed-integer nonlinear program, that is, even its continuous relaxation is nonconvex, since the power flow physics \eqref{eq:power_flow} are nonlinear and bilinear in the investment and operational variables. Off-the-shelf global solvers, notably spatial branch-and-bound implementations, solve such programs to a prescribed optimality tolerance when the variables are bounded \cite[Sec.~5]{belotti2013}. In practice, however, the effort of spatial branch-and-bound grows rapidly with the number of nonconvex terms, and global solvers do not scale to problems such as distribution system planning, which often feature thousands of continuous variables and hundreds of binary variables.

\eqref{eq:lodge} also fails to admit a faithful convex or linear surrogate. The classic branch flow model \cite{baran1989} and the associated second-order cone relaxation and linear LinDistFlow approximation \cite{low2014parti} assume a balanced network and are therefore only applicable to the trunk of the distribution network close to the substation. Gan and Low \cite{gan2014convex} extend the branch flow model to unbalanced three-phase networks, proposing a semidefinite relaxation and an associated Lin3DistFlow approximation, since extended to meshed networks \cite{wang2026linearized}. The semidefinite relaxation's only known sufficient conditions for exactness \cite{zhou2019sufficient} are prohibitively restrictive in practice, failing on every network studied here, and a mixed-integer semidefinite program presents major computational and implementation challenges. Lin3DistFlow's balanced-voltage approximation, combined with its omission of losses and linearized voltage drops, limits its fidelity on the heavily loaded and unbalanced feeders encountered in planning, where it can accept portfolios that violate limits under the full model (Section \ref{sec:results}). Other linearizations \cite{bernstein2017linear,bernstein2018load,girigoudar2021linearized} are unsuitable because they are taken around a fixed operating point, while their coefficients depend on the investment variables.

One-shot optimization, directly or through a surrogate, is therefore ill-suited to the three-phase DSEP problem \eqref{eq:lodge}. Instead, we separate its two sources of difficulty. We reinterpret and generalize the deficiency frameworks in \cite{overbye1994,jereminov2020,foster2022,zamzam2018} for planning by defining deficiency as a function of a candidate portfolio, evaluated under the full model \eqref{eq:power_flow}--\eqref{eq:inequalities} by standard OPF solvers, to isolate the nonlinear physics from the binary variables. We then propose a heuristic cutting-plane search in which a mixed-integer linear program uses the first-order information from each evaluation to propose the next portfolio. Section \ref{sec:theory} formalizes the notion of deficiency and develops the analysis underpinning the method, and Section \ref{sec:cutting} describes the search.

\section{Deficiency Framework} \label{sec:theory}

\ifarxiv
Proofs of the results in this section are given in Appendix \ref{app:proofs}.
\else
Proofs of the results in this section are given in the online appendix \cite{brock2026arxiv}.
\fi

\subsection{Deficiency} \label{sec:deficiency}

\eqref{eq:lodge} requires that any feasible portfolio $b$ admit a dispatch $t, B$ such that its constraints are satisfied with both slack families $\tilde{I}$ and $\sigma$ at zero. The same can be required of any monitored set $\mathcal{M} \subseteq \overline{\mathcal{M}}$ by imposing $g_{\mathcal{M}}$ in place of $g_{\overline{\mathcal{M}}}$.

\begin{definition}[Adequacy] \label{def:adequacy}
    A relaxed portfolio $b \in \operatorname{conv} \mathcal{B}$ is \textit{adequate for} $\mathcal{M} \subseteq \overline{\mathcal{M}}$ if there exists $x$ with $h(x, b) = 0$, $g_{\mathcal{M}}(x, b) \leq 0$, $\tilde{I} = \mathbf{0}$ and $\sigma = \mathbf{0}$.
\end{definition}

One can measure the degree of inadequacy of a given portfolio with respect to a given set of monitored components by fixing the portfolio, freeing either slack family $\tilde{I}$ or $\sigma$, and minimizing the slack required to restore feasibility. Definition \ref{def:deficiency} formalizes and generalizes this notion.

\begin{definition}[Deficiency variant] \label{def:deficiency}
    A collection of programs
    \begin{subequations} \label{eq:deficiency_program}
    \begin{align}
        v_{\mathcal{M}}^*(b) \equiv \min_{z} \quad & F_{\mathcal{M}}(z) \\
        \text{s.t.} \quad & \lambda: H_{\mathcal{M}}\left(z, b\right) = 0 \\
        & \mu: G_{\mathcal{M}}\left(z, b\right) \leq 0
    \end{align}
    \end{subequations}
    with Lagrange multipliers $\lambda$ and $\mu$ is called a \textit{deficiency variant} if, for every $\mathcal{M} \subseteq \overline{\mathcal{M}}$,
    \begin{conditions}
        \item \label{def:deficiency:contract} $v_{\mathcal{M}}^*(b) \geq 0$ for every $b \in \operatorname{conv} \mathcal{B}$, with equality if and only if $b$ is adequate for $\mathcal{M}$;
        \item \label{def:deficiency:polynomial} $F_{\mathcal{M}}$, $G_{\mathcal{M}}$ and $H_{\mathcal{M}}$ are polynomial in $(z, b)$;
        \item \label{def:deficiency:homogeneous} $F_{\mathcal{M}}$ is convex and positively homogeneous of a degree $s \geq 1$ common to every $\mathcal{M}$.
    \end{conditions}
    We call \eqref{eq:deficiency_program} the variant's \textit{deficiency program for} $\mathcal{M}$ and $v_{\mathcal{M}}^*$ its \textit{optimal value function}.
\end{definition}

Definition \ref{def:deficiency} supports equivalent representations of the 1-, 2- and $\infty$-norms over either slack family. Table \ref{tab:instances} provides the implementation of the five such combinations studied here, named for their slack family (C: compensation, E: exceedance) and norm, including the necessary epigraph variables $r$ and $\varsigma$ and monotonic objective transformations (the squared objectives of C2 and E2).\footnote{The $\infty$-norm compensation variant, C$\infty$, also satisfies Definition \ref{def:deficiency}; we omit it because its non-maximal compensations are unpenalized free inputs, which leaves the optimum degenerate and the local solver unreliable.} The compensation variants penalize the \textit{nominal compensation power}, given for phase $\phi \in \Phi_m$ of load $m \in \mathcal{D}_i$ by $\tilde{S}_{m,\phi} \equiv V_{i,\phi}^\dagger \tilde{I}_{m,\phi}^H$ and represented in real coordinates by $\tilde{r} \equiv (\operatorname{Re} \tilde{S}, \operatorname{Im} \tilde{S})$. The theory to follow applies to any single, fixed variant satisfying Definition \ref{def:deficiency}. We assume throughout that \eqref{eq:deficiency_program} is feasible at every $b$ in a neighborhood of $\operatorname{conv} \mathcal{B}$ within its affine hull, which holds trivially for the compensation family (compensation currents can shed all load) and for the exceedance family whenever the power flow physics admit a solution.

\begin{table}[!t]
\caption{Deficiency variants}
\label{tab:instances}
\centering
\renewcommand{\arraystretch}{1.2}
\begin{tabular}{lccllc}
\hline
 & $z$ & $F_{\mathcal{M}}(z)$ & $H_{\mathcal{M}}(z, b)$ & $G_{\mathcal{M}}(z, b)$ & $s$ \\
\hline
C2 & $x$ & $\lVert \tilde{r} \rVert_2^2$ & $h(x, b)$ & $g_{\mathcal{M}}(x, b)$ & 2 \\
 & & & $\sigma$ & & \\
\hline
C1 & $(x, r)$ & $\mathbf{1}^\top r$ & $h(x, b)$ & $g_{\mathcal{M}}(x, b)$ & 1 \\
 & & & $\sigma$ & $\pm \tilde{r} - r$ & \\
\hline
E2 & $x$ & $\lVert \sigma \rVert_2^2$ & $h(x, b)$ & $g_{\mathcal{M}}(x, b)$ & 2 \\
 & & & $\tilde{I}$ & & \\
\hline
E1 & $x$ & $\mathbf{1}^\top \sigma$ & $h(x, b)$ & $g_{\mathcal{M}}(x, b)$ & 1 \\
 & & & $\tilde{I}$ & & \\
\hline
E$\infty$ & $(x, \varsigma)$ & $\varsigma$ & $h(x, b)$ & $g_{\mathcal{M}}(x, b)$ & 1 \\
 & & & $\tilde{I}$ & $\sigma - \varsigma \mathbf{1}$ & \\
\hline
\end{tabular}
\end{table}

\subsection{Solution Branches} \label{sec:branches}

For the rest of this section, we fix a monitored set $\mathcal{M} \subseteq \overline{\mathcal{M}}$ and drop the subscript $\mathcal{M}$ from every quantity that depends on it. The assumption and results below apply to each $\mathcal{M}$; Section \ref{sec:cutting} restores the subscript where monitored sets are compared.

A deficiency program \eqref{eq:deficiency_program} is nonconvex in general and may have multiple local minimizers at a given $b$. With the portfolio fixed, however, even at a fractional $b \in \operatorname{conv} \mathcal{B}$, each program in Table \ref{tab:instances} is an OPF instance, in which $b$ enters only through the device parameters, and local solvers have been shown to reliably find near-optimal solutions of OPF problems \cite{pglib2019}. Assumption \ref{assumption:branch} formalizes the expectation that such a solver tracks a consistent local solution under a shifting portfolio, an investment-parametrized analog of the hypothesis behind continuation power flow \cite{ajjarapu1992}.

\begin{assumption}[Solution branch] \label{assumption:branch}
    There is a set $U \supseteq \operatorname{conv} \mathcal{B}$, open in the affine hull $\mathcal{A}$ of $\operatorname{conv} \mathcal{B}$, and a \textit{solution branch} $\zeta$ on $U$ such that
    \begin{conditions}
        \item \label{def:regular:minimizer} $\zeta(b)$ is a local minimizer of \eqref{eq:deficiency_program} at every $b \in U$;
        \item \label{def:regular:mfcq} the Mangasarian--Fromovitz constraint qualification holds at $\zeta(b)$ for every $b \in U$;
        \item \label{def:regular:semialgebraic} $\zeta$ is semialgebraic and locally Lipschitz.
    \end{conditions}
    With $L(z, \lambda, \mu, b) \equiv F(z) + \lambda^\top H(z, b) + \mu^\top G(z, b)$ the Lagrangian and
    \begin{align} \label{eq:multiplier_set}
        \Lambda(b) \equiv \big\{ (\lambda, \mu) :\; & \nabla_z L\left(\zeta(b), \lambda, \mu, b\right) = 0, \nonumber \\
        & \mu \geq 0,\; \mu^\top G\left(\zeta(b), b\right) = 0 \big\}
    \end{align}
    the set of Lagrange multipliers at $\zeta(b)$, we have oracle access to $\zeta(b)$ and to an arbitrary element of $\Lambda(b)$ at every $b \in \operatorname{conv} \mathcal{B}$.
\end{assumption}

Notice that the branch $\zeta$ is defined on a neighborhood of $\operatorname{conv} \mathcal{B}$ in $\mathcal{A}$, not only at the lattice points $\mathcal{B}$, which property \ref{def:deficiency:polynomial} of Definition \ref{def:deficiency} permits, so that every segment between portfolios lies in its domain.

Intuitively, Assumption \ref{assumption:branch} idealizes a local solver by $\zeta$: its Lipschitz requirement (condition \ref{def:regular:semialgebraic}) asserts that the solver tracks one sheet of local solutions, without basin-hopping, that the sheet does not fold (a saddle-node bifurcation) over the relaxed portfolio set, and that the minimizer does not jump across a flat face (as at a tie between two vertices of a linear program). Condition \ref{def:regular:mfcq}, the Mangasarian--Fromovitz qualification, holds for the slack-based programs in Table \ref{tab:instances}, because the slacks supply a direction that strictly relaxes every active limit. For the semialgebraic part of condition \ref{def:regular:semialgebraic}, the data of \eqref{eq:deficiency_program} are polynomial, so the local minimizers form finitely many semialgebraic sheets \cite[Prop.~2.2.4, Prop.~2.9.10]{bochnak1998} and a solver that tracks a single sheet returns a semialgebraic branch.

Since the optimal value function $v^*$ is inaccessible, we use the branch's value function $v \equiv F \circ \zeta$ as a proxy to classify portfolios as adequate or inadequate for $\mathcal{M}$. Assumption \ref{assumption:branch} posits that $v$ gives a consistent, continuous deficiency landscape over the relaxed portfolio set. Since the landscape lies above $v^*$, $v$ correctly classifies every inadequate portfolio, and any portfolio with $v(b) = 0$ is \textit{certified} as adequate for $\mathcal{M}$. Whether $v$ is a useful classifier depends on its tendency to report false negatives, that is, to fail to certify an adequate portfolio. Certifying every adequate portfolio of $\operatorname{conv} \mathcal{B}$ is a weaker requirement than global optimality, in particular because suboptimality at an inadequate $b$ still yields a correct classification. Moreover, false negatives are partially checkable: \textit{any} feasible point of \eqref{eq:deficiency_program} with zero objective certifies adequacy, so a practitioner can verify adequacy of some oversized portfolios by power-flow analysis under a heuristic dispatch rule and confirm that $\zeta$ certifies them, exposing any false negatives.

Proposition \ref{prop:sensitivity} shows that the multipliers provided along with $\zeta$ in Assumption \ref{assumption:branch} (and typically returned by off-the-shelf solvers such as Ipopt \cite{wachter2006}) also give the gradient of $v$.

\begin{proposition}[Sensitivity] \label{prop:sensitivity}
    Under Assumption \ref{assumption:branch}, $v$ is differentiable at every $b \in U$ outside a semialgebraic set of lower dimension than $U$, and at every $b$ at which it is differentiable,
    \begin{equation} \label{eq:sensitivity}
        \nabla v(b) = \frac{\partial H}{\partial b}\left(\zeta(b), b\right)^{\!\top}\!\lambda + \frac{\partial G}{\partial b}\left(\zeta(b), b\right)^{\!\top}\!\mu
    \end{equation}
    for every $(\lambda, \mu) \in \Lambda(b)$, with both sides read as linear functionals on the directions of $\mathcal{A}$.
\end{proposition}

Proposition \ref{prop:sensitivity} parallels but departs from classical sensitivity analysis \cite{fiacco1983,bertsekas2016nonlinear}, which assumes the linear-independence qualification, second-order sufficiency, and strict complementarity at a local solution to derive a unique, differentiable branch with a unique multiplier. However, this branch exists only in a neighborhood of the solution, and strict complementarity fails at active-set changes as the portfolio varies; comparing portfolios requires a consistent branch that spans the relaxed portfolio set. Assumption \ref{assumption:branch} therefore posits the branch and requires only the weaker Mangasarian--Fromovitz qualification, which is why \eqref{eq:sensitivity} holds for every multiplier and only where $v$ is differentiable.

The branch $\zeta$, then, allows us to evaluate a single portfolio at a time and access zero- and first-order information at that portfolio: how deficient it is and, through the gradient, how much relief each candidate offers, which a planner otherwise estimates by trial power flows. We therefore study the first-order approximation of the \textit{deficiency function} $\delta \equiv v^{1/s}$. Since $F$ is positively homogeneous of degree $s$, $\delta$ grows linearly with the slacks and shares their units. For the instances in Table \ref{tab:instances}, $\delta$ is the norm of the slacks itself rather than the square, so variants of the same slack family share a scale. By Proposition \ref{prop:sensitivity}, $v$ is differentiable except on a set of lower dimension than $U$. A portfolio in $\mathcal{B}$ falls in this lower-dimensional set only by coincidence, such as when a candidate is sized exactly to a binding limit; we assume that no such coincidence occurs. The first-order deficiency approximation at a portfolio $\bar{b} \in \mathcal{B}$ with $\delta(\bar{b}) > 0$ is
\begin{align*}
    \hat{\delta}_{\bar{b}}\left(b\right) & \equiv \delta\left(\bar{b}\right) + \nabla\delta\left(\bar{b}\right)^\top\left(b - \bar{b}\right) \\
    & = v\left(\bar{b}\right)^{1/s} + \tfrac{1}{s}\, v\left(\bar{b}\right)^{1/s - 1}\, \nabla v\left(\bar{b}\right)^\top\left(b - \bar{b}\right).
\end{align*}

\subsection{Approximation Error} \label{sec:validity}

Under Assumption \ref{assumption:branch}, fix a portfolio $\bar{b} \in \mathcal{B}$ with $\delta(\bar{b}) > 0$ and a portfolio $b \in \mathcal{B}$, $b \neq \bar{b}$, and let $d \equiv b - \bar{b}$. We are interested in the error incurred by estimating the deficiency of $b$ with the first-order approximation formed at $\bar{b}$. We define the approximation error $\Delta \equiv \hat{\delta}_{\bar{b}}(b) - \delta(b)$ and characterize it by analyzing the one-dimensional cross-section of $\delta$ between $\bar{b}$ and $b$. For clarity, we abuse notation throughout this subsection: every function defined over $\operatorname{conv} \mathcal{B}$ is reparameterized over $t \in [0, 1]$ along the segment, for example $\delta(t) \equiv \delta(\bar{b} + td)$.

We begin with two foundational lemmas. Define $\varphi \equiv F^{1/s}$, such that $\delta(t) = \varphi(\zeta(t))$. Lemma \ref{lem:phi_convex} establishes convexity of $\varphi$.

\begin{lemma}[Convexity] \label{lem:phi_convex}
    $\varphi$ is convex.
\end{lemma}

An \textit{event} is a time $t \in (0, 1)$ at which $\zeta(t)$ fails to be analytic or $\delta(t)$ has an isolated zero. For the programs in Table \ref{tab:instances}, an event is typically a change of the active set.

\begin{lemma}[Finitely many events] \label{lem:finite}
    Under Assumption \ref{assumption:branch}, there are finitely many events, and $\delta(t)$ and $\zeta(t)$ are Lipschitz on $[0, 1]$ with finite one-sided derivatives everywhere.
\end{lemma}

Lemma \ref{lem:finite} establishes that $\delta(t)$ is piecewise smooth: continuous on $[0, 1]$ and analytic except at the finitely many events $t_1 < \dots < t_K$, where its slope may jump by $J_i \equiv \delta'(t_i^+) - \delta'(t_i^-)$. Proposition \ref{prop:accounting} characterizes the approximation error in terms of these slope jumps and the curvature between events.

\begin{proposition}[Error accounting] \label{prop:accounting}
    Under Assumption \ref{assumption:branch},
    \begin{equation} \label{eq:accounting}
        \Delta = -\Delta_{\mathrm{A}} - \Delta_{\mathrm{N}} - \Delta_{\mathrm{R}} ,
    \end{equation}
    where the \textit{atomic term}, the \textit{non-atomic term} and the \textit{response} are
    \begin{align*}
        \Delta_{\mathrm{A}} & \equiv \sum_i (1 - t_i) J_i , \\
        \Delta_{\mathrm{N}} & \equiv \int_{\Omega} (1-\xi)\, \dot{\zeta}(\xi)^\top \nabla^2\varphi\left(\zeta(\xi)\right) \dot{\zeta}(\xi)\, \mathrm{d}\xi , \\
        \Delta_{\mathrm{R}} & \equiv \int_{\Omega} (1-\xi)\, \nabla\varphi\left(\zeta(\xi)\right)^\top \ddot{\zeta}(\xi)\, \mathrm{d}\xi ,
    \end{align*}
    with $\Omega \equiv \{t \in [0, 1] \setminus \{t_1, \dots, t_K\} : \delta(t) > 0\}$, and both integrals are finite.
\end{proposition}

Proposition \ref{prop:atomic} provides an expression for the atomic credit $J_i$ at each event.

\begin{proposition}[Atomic credit] \label{prop:atomic}
    Under Assumption \ref{assumption:branch}, $J_i \geq 0$ at every event. Where $v(t_i) > 0$,
    \begin{equation} \label{eq:atom_interior}
        J_i = \tfrac{1}{s}\, v(t_i)^{1/s - 1} \Big( \max_{\Lambda(t_i)} \partial_t L - \min_{\Lambda(t_i)} \partial_t L \Big) ,
    \end{equation}
    where $\partial_t L(\lambda, \mu, t) \equiv \partial_b L(\zeta(t), \lambda, \mu, \bar{b} + td)\, d$ and the extreme values are over $(\lambda, \mu) \in \Lambda(t_i)$; where $\delta(t_i) = 0$,
    \begin{equation} \label{eq:atom_boundary}
        J_i = \lvert \delta'(t_i^+) \rvert + \lvert \delta'(t_i^-) \rvert .
    \end{equation}
\end{proposition}

\begin{table}[!t]
\caption{Integrands of \eqref{eq:accounting} and boundary atoms by variant}
\label{tab:credit}
\centering
\renewcommand{\arraystretch}{1.3}
\begin{tabular}{lccc}
\hline
 & $\nabla\varphi\left(\zeta\right)^\top \ddot{\zeta}$ & $\dot{\zeta}^\top \nabla^2\varphi\left(\zeta\right) \dot{\zeta}$ & $J_i$, $\delta(t_i) = 0$ \\
\hline
E1 & $\mathbf{1}^\top \ddot{\sigma}$ & $0$ & $\mathbf{1}^\top \lvert \dot{\sigma}(t_i^-) \rvert$ \\
C1 & $\mathbf{1}^\top \ddot{r}$ & $0$ & $\mathbf{1}^\top \lvert \dot{r}(t_i^-) \rvert$ \\
E$\infty$ & $\ddot{\varsigma}$ & $0$ & $\lvert \dot{\varsigma}(t_i^-) \rvert$ \\
E2 & $\sigma^\top \ddot{\sigma} / \lVert \sigma \rVert$ & $\big\lVert \operatorname{proj}_{\sigma^\perp} \dot{\sigma} \big\rVert^2 / \lVert \sigma \rVert$ & $\lVert \dot{\sigma}(t_i^-) \rVert_2$ \\
C2 & $\tilde{r}^\top \ddot{\tilde{r}} / \lVert \tilde{r} \rVert$ & $\big\lVert \operatorname{proj}_{\tilde{r}^\perp} \dot{\tilde{r}} \big\rVert^2 / \lVert \tilde{r} \rVert$ & $\lVert \dot{\tilde{r}}(t_i^-) \rVert_2$ \\
\hline
\end{tabular}
\end{table}

Propositions \ref{prop:accounting} and \ref{prop:atomic} are best interpreted in the context of the instances in Table \ref{tab:instances}. Table \ref{tab:credit} gives, for each, the integrands of $\Delta_{\mathrm{R}}$ and $\Delta_{\mathrm{N}}$ and $J_i$ in the case \eqref{eq:atom_boundary}. The \textit{response} term, in the first column, is the acceleration of the slack trajectory along the segment: total under the 1-norms, of the maximal slack under the $\infty$-norm, and radial under the 2-norms. It is unsigned in general. If the branch responds approximately affinely to the portfolio along the segment, so that $\ddot{\zeta} \approx 0$, the response term is small.

Where the response term comes from the curvature of the trajectory $\zeta$, the non-atomic term in the second column comes from the curvature of $\varphi$. The non-atomic term is nonnegative since $\varphi$ is convex (Lemma \ref{lem:phi_convex}), and zero under the 1- and $\infty$-norms, where $\varphi$ is linear. For the 2-norm variants it is the squared transverse velocity of the slacks divided by their distance from the origin: a slack trajectory earns \textit{non-atomic credit} whenever it does not head directly toward or away from the origin, which is to be expected at each active-set change, where the slacks turn to serve a different set of limits.

Proposition \ref{prop:atomic} gives $J_i$ for two kinds of event: \eqref{eq:atom_interior} at \textit{interior events}, where the deficiency is positive, and \eqref{eq:atom_boundary} at \textit{boundary events}, where it is zero. $J_i$ is nonnegative by inspection in both cases, as Proposition \ref{prop:atomic} confirms, meaning that all slope jumps fold inward, in the convex direction. The third column of Table \ref{tab:credit} gives the boundary case. For all variants it is the speed at which the slack trajectory arrives at zero. Intuitively, for a segment on which the deficiency reaches and remains at zero, this quantity is simply the kink formed by the deficiency function transversally meeting zero. Therefore, any time $\delta(b) = 0$, the segment must reach zero, and it accrues \textit{atomic credit} whenever it meets zero transversally before $t = 1$. By \eqref{eq:atom_interior}, the atom at an interior event is positive only where the multiplier set $\Lambda(t_i)$ is not a single point, in particular at active-set changes where a group of multipliers is constrained only in sum. These include a component of $\sigma$ or $r$ reaching zero under E1 or C1 and a change of the binding epigraph row of $\varsigma$ under E$\infty$.

Taken together, \eqref{eq:accounting} weighs the unsigned response term $\Delta_{\mathrm{R}}$ against the two credit terms $\Delta_{\mathrm{A}}$ and $\Delta_{\mathrm{N}}$, both nonnegative by structure. Theorem \ref{thm:validity} formalizes this balance as a bound on the approximation error.

\begin{theorem}[Approximation error] \label{thm:validity}
    Under Assumption \ref{assumption:branch}, $\Delta_{\mathrm{A}} \geq 0$ and $\Delta_{\mathrm{N}} \geq 0$, so that
    \begin{equation} \label{eq:overestimation}
        \Delta \leq -\Delta_{\mathrm{R}} .
    \end{equation}
\end{theorem}

Therefore, the approximation tends to underestimate the deficiency, especially where the response is approximately affine ($\Delta_{\mathrm{R}} \approx 0$, as linearized power flow models assume) and the credit $\Delta_{\mathrm{A}} + \Delta_{\mathrm{N}}$ accumulated along the segment is appreciable. This tendency is strongest when $\delta(b) = 0$: the trajectory must then reach zero, which earns the boundary atom in Table \ref{tab:credit} whenever it arrives transversally, and typically crosses several active sets on the way, which earns both kinds of credit as argued above. Notice that every term of \eqref{eq:accounting} carries the weight $(1 - t_i)$, or $(1 - \xi)$ inside the integrals, so both credit and response have larger effects the earlier they occur along the segment. In particular, a deficiency that reaches zero early earns more credit.

Section \ref{sec:cutting} exploits Theorem \ref{thm:validity} to design a cutting-plane search with cuts given by the first-order approximations $\hat{\delta}_{\mathcal{M}, \bar{b}}$ for a range of monitored sets $\mathcal{M}$.

\section{Deficiency-Guided Search} \label{sec:cutting}

Given the idealized local solver of Assumption \ref{assumption:branch}, the deficiency yields a natural reformulation of \eqref{eq:lodge},
\begin{equation} \label{eq:branch_lodge}
    \min_{b \in \mathcal{B}} \; C(b) \quad \text{s.t.} \quad \delta_{\overline{\mathcal{M}}}(b) \leq 0 .
\end{equation}
Since the deficiency is nonnegative, the inequality constraint is equivalent to equality. The feasible set of \eqref{eq:branch_lodge} is therefore the set of portfolios $b \in \mathcal{B}$ with $\delta_{\overline{\mathcal{M}}}(b) = 0$, a subset of the feasible set of \eqref{eq:lodge}. The two problems have the same optimal value if and only if $\zeta_{\overline{\mathcal{M}}}$ certifies some optimal portfolio of \eqref{eq:lodge}; otherwise \eqref{eq:branch_lodge} is a restriction.

Adequacy for $\overline{\mathcal{M}}$ implies adequacy for every $\mathcal{M} \subseteq \overline{\mathcal{M}}$. Provided that each $\zeta_{\mathcal{M}}$ certifies every portfolio $\zeta_{\overline{\mathcal{M}}}$ certifies, \eqref{eq:branch_lodge} has the same feasible set as
\begin{equation} \label{eq:branch_lodge_all}
    \min_{b \in \mathcal{B}} \; C(b) \quad \text{s.t.} \quad \delta_{\mathcal{M}}(b) \leq 0 \quad \forall \, \mathcal{M} \subseteq \overline{\mathcal{M}} .
\end{equation}

\eqref{eq:branch_lodge_all} cannot be solved directly, because the deficiency functions $\delta_{\mathcal{M}}$ have no closed form. However, under Assumption \ref{assumption:branch} and Proposition \ref{prop:sensitivity}, we can access their zero- and first-order information at a given portfolio. Suppose $\zeta_{\mathcal{M}}$ has been evaluated at $\bar{b}$ for every pair $(\bar{b}, \mathcal{M})$ in a finite set $\mathcal{P}$, each with $\delta_{\mathcal{M}}(\bar{b}) > 0$. Imposing the first-order approximations of Section \ref{sec:branches} in place of the constraints of \eqref{eq:branch_lodge_all} gives the \textit{master problem}
\begin{subequations} \label{eq:inv-subproblem}
    \begin{align}
        \min_{b} \quad & C(b) \\
        \label{eq:cut} \text{s.t.} \quad & \hat{\delta}_{\mathcal{M}, \bar{b}}\left(b\right) \leq 0 \quad \forall \, \left(\bar{b}, \mathcal{M}\right) \in \mathcal{P} \\
        & b \in \mathcal{B}
    \end{align}
\end{subequations}
\eqref{eq:inv-subproblem} is a mixed-integer linear program, which off-the-shelf solvers solve to global optimality. The master problem is motivated by Theorem \ref{thm:validity}, which establishes that the cut \eqref{eq:cut} of a pair $(\bar{b}, \mathcal{M})$ excludes a $b$ feasible in \eqref{eq:branch_lodge} only if the response along the segment from $\bar{b}$ to $b$ is concave enough to overwhelm the buffer provided by the nonnegative credit terms. Moreover, as discussed in Section \ref{sec:validity}, these credit terms are most assured when $\delta_{\mathcal{M}}(b) = 0$, the only case of interest here: overestimation at a $b$ with $\delta_{\mathcal{M}}(b) > 0$ correctly excludes $b$ and is therefore inconsequential. Admitting such a $b$ is similarly benign, since it is cut in turn if proposed.

\eqref{eq:inv-subproblem} suggests a simple heuristic cutting-plane search that alternately solves the master problem for a proposal and evaluates the deficiency $\delta_{\overline{\mathcal{M}}}$ for the full monitored set at that proposal, adding the resulting \textit{aggregate cut} \eqref{eq:cut}. In this straightforward form, each cut reveals only a scalar deficiency, and the sensitivity of each candidate is the sum of the relief it provides over \textit{all} violated limits. In particular, the master problem cannot connect specific candidates to specific violations and may attempt to address a violation by building candidates that relieve other limits. We observe that this cross-substitution can cause iteration counts to grow combinatorially, and that it is effectively addressed by \textit{targeted cuts} for restricted monitored sets $\mathcal{M} \subsetneq \overline{\mathcal{M}}$ that isolate a few of the violated limits. A targeted cut requires, to first order, that every later proposal relieve its targeted limits. See Section \ref{sec:results} for further discussion and motivation.

As a means of identifying promising monitored subsets, we introduce the metric of a \textit{share}, which characterizes the salience of the components in $\mathcal{M}$ with respect to the deficiency solution for $\overline{\mathcal{M}}$. For the variants in Table \ref{tab:instances}, whose programs contain the operating limits \eqref{eq:inequalities}, let $\delta(b; \theta)$ denote the deficiency function for $\overline{\mathcal{M}}$ with $\sigma$ replaced by $\sigma + \theta$ in \eqref{eq:inequalities}, so that $\delta(b; \mathbf{0}) = \delta_{\overline{\mathcal{M}}}(b)$. Define
\begin{equation} \label{eq:share}
    \operatorname{share}_{e}(b) \equiv -\sum_{k \in \mathcal{K}_e} u_k \frac{\partial \delta}{\partial \theta_k}\left(b; \mathbf{0}\right)
\end{equation}
for an increment $u \geq \mathbf{0}$ indexed like $\sigma$, where $\mathcal{K}_e$ collects the entries of $\sigma$ belonging to the limits of component $e \in \overline{\mathcal{M}}$, and let $\operatorname{share}_{\mathcal{M}}(b) \equiv \sum_{e \in \mathcal{M}} \operatorname{share}_{e}(b)$. Intuitively, the share is the sensitivity of the deficiency to an incremental loosening of the limits of the components of $\mathcal{M}$ along the direction $u$. For the exceedance variants E1, E2 and E$\infty$ we use $u = \sigma^*$, the exceedances at $\zeta_{\overline{\mathcal{M}}}(b)$, to capture how much loosening each limit needs before its exceedance vanishes; this information is not available for the compensation variants C1 and C2, so we use $u = \mathbf{1}$. By the argument of Proposition \ref{prop:sensitivity}, with $\theta$ in place of $b$, the share is read from $\zeta_{\overline{\mathcal{M}}}(b)$ and its multipliers without a further solve; Section \ref{sec:implementation} gives the derivative.

Algorithm \ref{alg:cp-lodge} formalizes the cutting-plane search built around \eqref{eq:inv-subproblem}, augmented with targeted cuts identified by share.

\begin{algorithm}
\caption{CP-DSEP} \label{alg:cp-lodge}
    $\mathcal{P} \gets \emptyset$ \;
    \Repeat{$\delta_{\overline{\mathcal{M}}}\left(b^*\right) \leq 0$}{
    $b^* \gets$ solution of \eqref{eq:inv-subproblem} with cuts $\mathcal{P}$ \label{step:ip} \;
    $z^* \gets \zeta_{\overline{\mathcal{M}}}\left(b^*\right)$ \label{step:opf} \;
    \If{$\delta_{\overline{\mathcal{M}}}\left(b^*\right) > 0$ \label{step:check}}{
        $\mathcal{P} \gets \mathcal{P} \cup \left\{\left(b^*, \overline{\mathcal{M}}\right)\right\}$ \label{step:append} \;
        $e_1, \dots, e_m \gets$ the components $e \in \overline{\mathcal{M}}$ with $\operatorname{share}_e(b^*) > 0$, ordered so that $\operatorname{share}_{e_1}(b^*) \geq \dots \geq \operatorname{share}_{e_m}(b^*)$ \label{step:rank} \;
        \For{$j = 1, \dots, m$}{
            $\mathcal{M} \gets \left\{e_1, \dots, e_j\right\}$ \;
            \If{$\delta_{\mathcal{M}}\left(b^*\right) > 0$ \label{step:restricted}}{
                $\mathcal{P} \gets \mathcal{P} \cup \left\{\left(b^*, \mathcal{M}\right)\right\}$; \textbf{break} \label{step:targeted} \;
            }
        }
    }
    }
    \Return $\left(b^*, z^*\right)$
\end{algorithm}

For each proposed portfolio, Algorithm \ref{alg:cp-lodge} adds an aggregate cut in Step \ref{step:append} and then attempts to add a targeted cut by adding components to the monitored set in order of decreasing share, thereby isolating a small number of components for which the proposed portfolio is deficient even when all other components are relaxed. If no such monitored set is found, no targeted cut is added.

Since each cut excludes its own portfolio, Algorithm \ref{alg:cp-lodge} terminates with a certified portfolio, which is feasible in \eqref{eq:lodge}, or with an infeasible master problem in Step \ref{step:ip}. The latter outcome materializes only if $\delta_{\overline{\mathcal{M}}}(b) > 0$ for every $b \in \mathcal{B}$ or if, in the language of Section \ref{sec:validity}, the response is concave enough to outweigh the credit so pervasively across the deficiency landscape that the cuts exclude \textit{every} portfolio certified by $\zeta_{\overline{\mathcal{M}}}$.

\section{Numerical Experiments} \label{sec:experiments}

\subsection{Case Study} \label{sec:case_study}

We evaluate the search on six 12.47 kV distribution feeders from the SMART-DS synthetic dataset \cite{smartds2021}, three from its Austin region (\texttt{p1rhs7}, \texttt{p1uhs2} and \texttt{p3uhs2}) and three from its San Francisco region (\texttt{p5rhs0}, \texttt{p11uhs0} and \texttt{p15uhs4}). Each feeder is run in two of the dataset's published 2018 scenarios: the base scenario, which contains no distributed generation, and the highest customer photovoltaic adoption scenario (\texttt{solar\_extreme}), which places PV systems, sized by building footprint, at 85\% of customer locations. All loads are set to their maximum demand over the annual profiles provided in the data, and each PV system injects its rated DC capacity at unity power factor as a noncurtailable negative load. Since the networks are sized for coincident peaks \cite{smartds2021}, setting every individual peak and rating at once produces moderate violations that call for a portfolio of targeted upgrades.

Line charging and transformer magnetizing branches are neglected. Single-phase center-tapped service transformers are folded into loss-preserving equivalent loads on the primary side, removing the low-voltage secondaries at the leaves of the network. Delta-connected service transformers, which appear only on \texttt{p5rhs0}, are converted into exact wye-wye equivalents.

\cite{kersting2017} provides the impedances and ampacities of seven standard ACSR sizes, and the distribution system upgrade unit cost database \cite{nrel2017costdb, horowitz2018cost} provides PG\&E's per-unit-length cost of each, interpreted as three-phase construction and scaled by conductor count on lines with fewer phases; an eighth rung represents two ACSR 477 circuits in parallel, at twice the cost and ampacity and half the impedance, with mutual coupling neglected. The database also provides six standard transformer sizes from 150 to 2500 kVA with SDG\&E's unit costs; each transformer candidate keeps the existing transformer's percent impedance on its own rating. The menu $\mathcal{G}_{ij}$ of every $i \to j \in \mathcal{L} \cup \mathcal{T}$ comprises the existing device and every standard candidate of greater capacity. All costs are the database's 2016 figures, unescalated.

The network graph, less the substation, is partitioned into five zones by spectral clustering on its unweighted adjacency matrix, and each zone receives at most one candidate regulator, on the line supplying the zone's bus nearest the substation. Candidate regulators are rated 10 MVA, matching the dataset's existing units, and cost \$166\,000, the mean total installed cost of the database's 12--12.7\,kV line regulators. All regulators, existing and candidate, are assigned a 1\% resistive impedance and a $\pm 10\%$ tap range. A 300\,kVAr three-phase capacitor bank candidate, costing \$10\,700 \cite{horowitz2018cost}, is placed at every three-phase bus at primary voltage other than the substation bus and buses already carrying a bank. Installed banks, existing or new, must lie more than four branches apart, a linear spacing constraint implemented as an instance of the generic constraint \eqref{eq:side}. As implemented, the instances of \eqref{eq:lodge} range from 151 to 439 buses, 757 to 2\,352 binary variables and 1\,777 to 6\,015 real continuous variables, excluding the slacks.

\subsection{Linear Benchmark} \label{sec:linear}

We compare the portfolios of Algorithm \ref{alg:cp-lodge} against a mixed-integer second-order cone program generated by applying the Lin3DistFlow approximation of \cite{gan2014convex} to the radial networks of the study. The complex power $S_{ij,\phi}$ on each branch $i \to j \in \mathcal{E}$ is split over its menu as $S_{ij,\phi} = \sum_{k \in \mathcal{G}_{ij}} S_{ij,\phi}^k$, with each candidate's share limited by $\lvert S_{ij,\phi}^k \rvert \leq b_k \overline{S}_k$, so that by \eqref{eq:single_upgrade} only the installed candidate carries flow. For a transformer candidate $k \in \mathcal{G}_{ij}$, $i \to j \in \mathcal{T}$, $\overline{S}_k$ is used directly; for a line candidate, $\overline{S}_k \equiv \lvert V_{i,\phi}^\dagger \rvert \overline{I}_k$; and for a regulator candidate, $\overline{S}_k$ is chosen large enough that it does not bind. Similarly, a capacitor bank $m \in \mathcal{C}_i$ enters the power balance, as in \cite{gan2014convex}, as a reactive power injection, approximated at nominal voltage as $B_m \lvert V_{i,\phi}^\dagger \rvert^2$. The voltage drops \eqref{eq:line_ohms}, \eqref{eq:xfmr_ohms} and \eqref{eq:reg_ohms} become
\begin{subequations} \label{eq:linear_power_flow}
\begin{constraintblock}
    \begin{flalign}
        \label{eq:lin_line_ohms} & v_i = v_j + \sum_{k \in \mathcal{G}_{ij}} 2\operatorname{Re}(\tilde{Z}_k S_{ij}^k) && \forall i \to j \in \mathcal{L}
    \end{flalign}
    \begin{flalign}
        \nonumber & v_{i,\phi} = n_{ij}^2 v_{j,\phi} + \sum_{k \in \mathcal{G}_{ij}} 2\operatorname{Re}(Z_k S_{ij,\phi}^k) && \\
        \label{eq:lin_xfmr_ohms} & \qquad \forall i \to j \in \mathcal{T}, \phi \in \Phi_j
    \end{flalign}
    \begin{flalign}
        \nonumber & v_{i,\phi} = u_{ij,\phi} + \sum_{k \in \mathcal{G}_{ij}} 2\operatorname{Re}(Z_k S_{ij,\phi}^k) && \\
        \label{eq:lin_reg_ohms} & \qquad \forall i \to j \in \mathcal{R}, \phi \in \Phi_j
    \end{flalign}
\end{constraintblock}
\end{subequations}
where $v_{i,\phi}$ is the squared voltage magnitude, $Z_k \equiv Y_k^{-1}$, $\tilde{Z}_k \equiv \text{Diag}(\alpha)(Z_k^{H})^{\top}\text{Diag}(\alpha)^H$ is the balanced-voltage counterpart of a line candidate's impedance with $\alpha \equiv (1, e^{-\mathbf{j}2\pi/3}, e^{\mathbf{j}2\pi/3})$, and $u_{ij,\phi} \equiv t_{ij,\phi}^{-2} v_{j,\phi}$ is a regulator's post-tap squared voltage, whose tap range \eqref{eq:tap_range} becomes
\begin{flalign}
    \nonumber & b_k \implies \left(1 + \tau_k\right)^{-2}v_{j,\phi} \leq u_{ij,\phi} \leq \left(1 - \tau_k\right)^{-2}v_{j,\phi} && \\
    \label{eq:squared_tap_indicator} & \qquad \forall i \to j \in \mathcal{R}, \phi \in \Phi_j, k \in \mathcal{G}_{ij} .
\end{flalign}
The benchmark minimizes $C(b)$ over $b \in \mathcal{B}$ subject to these constraints, the power balance and slack-bus conditions of Lin3DistFlow, the voltage limits \eqref{eq:volt_lims} on the squared voltages, and the capacitor bound \eqref{eq:capacitor_bound}.

\subsection{Implementation} \label{sec:implementation}

In place of the idealized local solver $\zeta_{\mathcal{M}}$, we solve \eqref{eq:deficiency_program} with Ipopt \cite{wachter2006} (tolerance $10^{-8}$, linear solver MA57 \cite{hsl}, gradient-based scaling, no acceptable-convergence termination). Every solve for $\overline{\mathcal{M}}$ starts from the same cold point (flat nominal voltages, regulators at neutral, capacitors at zero susceptance, slacks inside their limits), so the solver is a function of the portfolio. The solves for $\mathcal{M} \subsetneq \overline{\mathcal{M}}$ are warm-started at $\zeta_{\overline{\mathcal{M}}}(b^*)$, which is always available in Algorithm \ref{alg:cp-lodge} and already has objective zero wherever $\zeta_{\overline{\mathcal{M}}}$ certifies $b^*$, as Section \ref{sec:cutting} requires. The do-not-invest proxies are implemented at $\varepsilon = 0$, with each pinned device's coinciding bounds replaced by an equality whose multiplier is projected on the direction in which the bound can relax, the limit of the multiplier as $\varepsilon \to 0$. The master problem \eqref{eq:inv-subproblem} and the benchmark are solved with Gurobi \cite{gurobi}, the master to a relative optimality gap of 0.3\%, which bounds each returned cost within 0.3\% of the optimum of its final master. All arms ran sequentially on a laptop with an Intel Core i7-12800H processor (14 cores, 20 threads) and 16\,GB of memory, with models built in Pyomo \cite{bynum2021pyomo} and solved by Ipopt 3.14.20 and Gurobi 12.0.3.

The derivative in \eqref{eq:share} is $\partial \delta / \partial \theta_k (b; \mathbf{0}) = \tfrac{1}{s}\, v_{\overline{\mathcal{M}}}(b)^{1/s - 1}\, \mu_g^\top\, \partial g_{\overline{\mathcal{M}}} / \partial \sigma_k$, evaluated at $(\zeta_{\overline{\mathcal{M}}}(b), b)$, where $\mu_g$ collects the multipliers of the rows of $g_{\overline{\mathcal{M}}}$, each of which is a row of $G_{\overline{\mathcal{M}}}$ for the variants in Table \ref{tab:instances}.

\subsection{Results} \label{sec:results}

We run Algorithm \ref{alg:cp-lodge} under all five deficiency variants in Table \ref{tab:instances} on each of the twelve feeder--scenario pairs of Section \ref{sec:case_study}, for a total of 60 arms. The master problem remained feasible throughout, so Algorithm \ref{alg:cp-lodge} returned a certified portfolio on every arm. Table \ref{tab:performance} reports the iterations and wall clock of each arm, with the fastest variant on each case in bold. The five variants return portfolios of equal cost on eleven of the twelve cases, differing only in capacitor bank siting; on \texttt{base/p11uhs0} the C2 portfolio costs one capacitor bank, or 0.16\%, more. Figure \ref{fig:cost} breaks down the cost by device type of the portfolio shared by all variants (all but C2 on \texttt{base/p11uhs0}), beside that of the linear benchmark.

\begin{table}[!t]
\caption{Iterations and wall clock (seconds) by variant}
\label{tab:performance}
\centering
\footnotesize
\setlength{\tabcolsep}{4pt}
\begin{tabular}{lr|rrrrr}
\hline
feeder & Lin3DF & C1 & C2 & E1 & E2 & E$\infty$ \\
\hline
\multicolumn{7}{l}{\textit{base}} \\
p11uhs0 & 6.1 & 45/271.5 & 44/210.0 & 67/372.0 & 70/429.0 & \textbf{43/163.2} \\
p15uhs4 & 6.3 & 12/72.1 & 11/62.4 & 16/90.5 & 12/74.8 & \textbf{11/49.4} \\
p1rhs7 & 2.9 & 2/3.5 & 2/3.4 & 2/3.5 & 2/3.5 & \textbf{2/3.0} \\
p1uhs2 & 6.2 & 16/111.3 & 15/87.0 & 16/94.8 & 15/90.3 & \textbf{14/65.5} \\
p3uhs2 & 2.1 & 3/4.9 & 2/3.0 & 2/3.0 & 2/3.0 & \textbf{2/2.6} \\
p5rhs0 & 2.7 & 1/1.6 & \textbf{1/1.4} & 1/1.5 & \textbf{1/1.4} & \textbf{1/1.4} \\
\multicolumn{7}{l}{\textit{pv}} \\
p11uhs0 & 5.3 & 3/11.1 & 4/14.1 & \textbf{3/9.5} & 3/10.5 & 5/14.1 \\
p15uhs4 & 6.8 & 19/113.6 & 16/88.4 & 16/88.7 & 16/89.8 & \textbf{16/68.3} \\
p1rhs7 & 2.7 & 2/3.6 & 2/3.6 & 2/3.6 & 2/3.5 & \textbf{2/3.0} \\
p1uhs2 & 6.2 & 4/20.6 & \textbf{4/20.1} & 5/25.5 & 5/25.9 & 5/21.0 \\
p3uhs2 & 2.2 & 1/1.3 & 1/1.3 & 1/1.3 & \textbf{1/1.2} & \textbf{1/1.2} \\
p5rhs0 & 2.9 & \textbf{17/49.6} & 19/52.2 & 25/69.5 & 26/73.2 & 26/55.3 \\
\hline
\textit{total} & 52.4 & 664.7 & 546.9 & 763.4 & 806.1 & 448.0 \\
\hline
\end{tabular}

\end{table}

\begin{figure}[!t]
\centering
\includegraphics[width=\columnwidth]{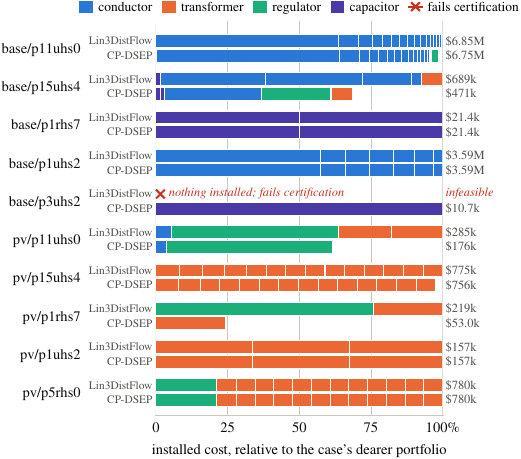}
\caption{Installed cost per case; the two cases that need no upgrades are omitted.}
\label{fig:cost}
\end{figure}

Ipopt reported a locally optimal solution on every call in the study, with no restoration failure, iteration-limit termination or detected infeasibility. The search also depends on the solver certifying adequate portfolios, which we cannot confirm in general. The variants provide some empirical validation: their programs and branches differ, yet they certify portfolios of equal cost on eleven of the twelve cases, so a false negative would have to be shared by all five branches to go unseen. Moreover, on no case does the benchmark return a cheaper feasible portfolio than the certified ones.

Section \ref{sec:validity} predicts that the first-order approximation around a given portfolio rarely overestimates the deficiency. We test this prediction by evaluating, within each run, every cut at the earlier portfolios where cuts of the same value function were generated. In all 9\,937 such checks the cut lies below the deficiency, strong empirical support for this prediction. On \texttt{base/p11uhs0}, an aggregate C2 cut and the targeted cut generated beside it, at one portfolio, exclude all four one-bank-cheaper portfolios certified by the other variants, by up to 6.7\% and 5.6\% of their deficiencies, respectively. These are the only overestimates found.

We test each of the benchmark's portfolios with the E$\infty$ deficiency, which every portfolio of the search also passes. Ten of the twelve cases require upgrades. The benchmark finds the global optimum of an approximate model, so Algorithm \ref{alg:cp-lodge} can outperform it only if two conditions hold. First, linearization must lose enough physical fidelity to change the set of feasible or optimal portfolios, as it does here: on \texttt{base/p3uhs2}, the linear model installs nothing because it places the lowest voltage at 0.9503 pu against a true 0.9455, below the 0.95 pu limit (a change in feasibility), and on five of the other nine cases the search certifies a cheaper portfolio (a change in optimality).
Second, the properties identified in Theorem \ref{thm:validity} must be reliable enough that the cuts of a simple sensitivity-based search do not exclude the portfolios that improve on the linearized optimum, as the same five cases confirm. In particular, on \texttt{base/p11uhs0} and \texttt{base/p15uhs4} the search trades conductor capacity for voltage control, exemplifying the substitutions that the sensitivity information reveals to a planner.

The benchmark solves every case in under 7\,s, 52\,s in total. The search takes 448\,s in total under E$\infty$, the fastest variant on seven of the twelve cases, and takes 547 to 806\,s under the others. No arm needs more than 70 proposals and 52 of the 60 need fewer than 20; each proposal that receives an aggregate cut needs between one and six further solves to search for a targeted cut, and one for 501 of the 625. With aggregate cuts alone, several arms crawl through the cross-substitution described in Section \ref{sec:cutting}, some to a 400-iteration limit, which motivates the targeted cuts.

\section{Conclusion} \label{sec:conclusion}

This work proposed an iterative, deficiency-guided cutting-plane method for unbalanced three-phase distribution system planning. The method mirrors existing planning loops, with an OPF in place of the power flow study and a mixed-integer linear program in place of manual upgrade selection, and each subproblem is handled by mature off-the-shelf solvers. At every iteration, the deficiency locates violations and its sensitivities estimate the relief each candidate would provide, information a planning engineer can inspect directly. We showed that the error of the resulting first-order approximation is bounded above by the concavity of the slack response, so that the approximation tends to underestimate the deficiency, which motivates the search. The search returned a certified portfolio on all 60 arms of Section \ref{sec:experiments}. Its portfolios outperformed those of a linearized benchmark, at lower cost on five of the ten cases that need upgrades, and its five variants agreed in cost on eleven of twelve cases. These results show that a deficiency-guided search can bring optimization-based planning to unbalanced three-phase feeders under the full network physics.
\ifarxiv
\appendices
\section{Proofs of the Results of Section \ref{sec:theory}} \label{app:proofs}

Throughout, $\lVert \cdot \rVert$ is the Euclidean norm. Two elementary facts about semialgebraic functions of one variable are used repeatedly. First, every semialgebraic subset of $\mathbb{R}$ is a finite union of points and intervals. Second, let $f$ be continuously differentiable and bounded, $\lvert f \rvert \leq c$, on an open interval $I$, with $f'$ semialgebraic. Then the sets $\{f' > 0\}$ and $\{f' < 0\}$ are finite unions of intervals, so $I$ splits into finitely many subintervals, say $m$, on each of which $f$ is monotone, and $\int_I \lvert f' \rvert \leq 2cm$. In particular, $f'$ is integrable on $I$. Derivatives of semialgebraic functions, partial derivatives included, are semialgebraic, since their graphs are defined by first-order formulas \cite[Prop.~2.2.4]{bochnak1998} (see \cite[Prop.~2.9.1]{bochnak1998} for one variable), and so are compositions of semialgebraic maps.

A \textit{direction} of $\mathcal{A}$ is a difference $b - b'$ of two of its points. Since $U$ is open in $\mathcal{A}$, $b + \tau d \in U$ for every $b \in U$, every direction $d$ and every sufficiently small $\tau > 0$. Lemma \ref{lem:directional}, on which Propositions \ref{prop:sensitivity} and \ref{prop:atomic} both rest, computes the one-sided directional derivatives of the branch's value. As in Sections \ref{sec:branches} and \ref{sec:validity}, $\mathcal{M}$ is fixed and its subscript dropped.

\begin{lemma}[Directional derivatives] \label{lem:directional}
    Under Assumption \ref{assumption:branch}, at every $b \in U$ and in every direction $d$ of $\mathcal{A}$, the one-sided directional derivative $v'(b; d) \equiv \lim_{\tau \downarrow 0} (v(b + \tau d) - v(b)) / \tau$ exists and equals
    \begin{equation*}
        v'(b; d) = \max_{(\lambda, \mu) \in \Lambda(b)} \partial_b L\left(\zeta(b), \lambda, \mu, b\right) d .
    \end{equation*}
\end{lemma}
\begin{proof}
    By conditions \ref{def:regular:minimizer} and \ref{def:regular:mfcq} of Assumption \ref{assumption:branch}, $\Lambda(b)$ is nonempty and compact \cite{gauvin1977}, so the maximum is attained. Write $b_\tau \equiv b + \tau d$ and $\ell(\lambda, \mu, b') \equiv \partial_b L(\zeta(b'), \lambda, \mu, b')\, d$, and let $\kappa_\zeta$ be a Lipschitz constant of $\zeta$ near $b$. The Lagrangian is a polynomial in $(z, b)$ whose coefficients are affine in $(\lambda, \mu)$.

    \textit{Lower bound.} Let $(\lambda, \mu) \in \Lambda(b)$ and $\tau > 0$. The point $\zeta(b_\tau)$ is feasible at $b_\tau$ and $\mu \geq 0$, so $v(b_\tau) = F(\zeta(b_\tau)) \geq L(\zeta(b_\tau), \lambda, \mu, b_\tau)$, while $v(b) = L(\zeta(b), \lambda, \mu, b)$ by feasibility and complementarity. The gradient $\nabla_z L$ vanishes at $(\zeta(b), \lambda, \mu, b)$ by stationarity, and $\lVert \zeta(b_\tau) - \zeta(b) \rVert \leq \kappa_\zeta \tau \lVert d \rVert$. A second-order Taylor expansion of $L$ about $(\zeta(b), b)$ therefore gives
    \begin{equation*}
        v(b_\tau) - v(b) \geq \tau\, \ell(\lambda, \mu, b) - \kappa \tau^2 ,
    \end{equation*}
    with a constant $\kappa$ uniform over $(\lambda, \mu)$ in any bounded set. Dividing by $\tau$ and letting $\tau \downarrow 0$, the lower limit of the difference quotients is at least $\ell(\lambda, \mu, b)$ for every $(\lambda, \mu) \in \Lambda(b)$, and hence at least $\max_{\Lambda(b)} \ell$.

    \textit{Upper bound.} Let $(\lambda_\tau, \mu_\tau) \in \Lambda(b_\tau)$. These multipliers remain bounded as $\tau \downarrow 0$. Otherwise, along a sequence with $\lVert (\lambda_\tau, \mu_\tau) \rVert \to \infty$, the normalized multipliers would converge to a nonzero $(\hat{\lambda}, \hat{\mu})$ with $\hat{\mu} \geq 0$, $\nabla_z H^\top \hat{\lambda} + \nabla_z G^\top \hat{\mu} = 0$ at $(\zeta(b), b)$, and $\hat{\mu}_i = 0$ for every constraint inactive at $\zeta(b)$, by the continuity of $\zeta$ and of the constraint functions; such a vector cannot exist under the Mangasarian--Fromovitz qualification at $\zeta(b)$, condition \ref{def:regular:mfcq}. By the same continuity, every limit point of the multipliers satisfies the conditions defining $\Lambda(b)$. Exchanging the roles of $b$ and $b_\tau$ in the argument above, and expanding about $(\zeta(b_\tau), b_\tau)$,
    \begin{equation*}
        v(b_\tau) - v(b) \leq \tau\, \ell(\lambda_\tau, \mu_\tau, b_\tau) + \kappa' \tau^2 ,
    \end{equation*}
    with $\kappa'$ uniform because the multipliers are bounded. Take $\tau_n \downarrow 0$ along which the difference quotients approach their upper limit and the multipliers converge to some $(\bar{\lambda}, \bar{\mu}) \in \Lambda(b)$. Since $\ell$ is continuous, the upper limit is at most $\ell(\bar{\lambda}, \bar{\mu}, b) \leq \max_{\Lambda(b)} \ell$.

    The two bounds show that the limit exists and equals $\max_{\Lambda(b)} \ell$.
\end{proof}

\begin{proof}[Proof of Proposition \ref{prop:sensitivity}]
    An affine bijection from $\mathbb{R}^m$, $m = \dim \mathcal{A}$, onto $\mathcal{A}$ preserves semialgebraicity, openness, Lipschitz continuity and dimension, so $U$ may be regarded as an open subset of $\mathbb{R}^m$. By condition \ref{def:regular:semialgebraic} of Assumption \ref{assumption:branch}, $\zeta$ is semialgebraic, and so is its domain $U$. The $C^1$ cell decomposition theorem \cite[Ch.~7, (3.2)]{vandendries1998}, applied to each component of $\zeta$ and followed by a common refinement, therefore partitions $U$ into finitely many cells, on each of which $\zeta$ is continuously differentiable. The union of the cells of lower dimension than $U$ is a semialgebraic set of lower dimension than $U$, and on each remaining, open cell, $v = F \circ \zeta$ is continuously differentiable.

    Let $v$ be differentiable at $b$, and let $d$ be a direction of $\mathcal{A}$. By Lemma \ref{lem:directional}, applied in the directions $d$ and $-d$,
    \begin{align*}
        \nabla v(b)^\top d & = v'(b; d) = \max_{\Lambda(b)} \partial_b L\, d , \\
        -\nabla v(b)^\top d & = v'(b; -d) = -\min_{\Lambda(b)} \partial_b L\, d ,
    \end{align*}
    with $\partial_b L$ evaluated at $(\zeta(b), \lambda, \mu, b)$. Hence $\partial_b L\, d$ takes the single value $\nabla v(b)^\top d$ over $\Lambda(b)$. Since $F$ does not depend on $b$, $\partial_b L\, d = (\partial_b H^\top \lambda + \partial_b G^\top \mu)^\top d$, and as $d$ is an arbitrary direction of $\mathcal{A}$, \eqref{eq:sensitivity} holds for every $(\lambda, \mu) \in \Lambda(b)$.
\end{proof}

\begin{proof}[Proof of Lemma \ref{lem:phi_convex}]
    For $s = 1$, $\varphi = F$, which is convex by property \ref{def:deficiency:homogeneous} of Definition \ref{def:deficiency}. Let $s > 1$. For any $z$, the restriction $\kappa \mapsto F(\kappa z) = \kappa^s F(z)$ of $F$ to the ray $\kappa > 0$ is convex, which for $s > 1$ forces $F(z) \geq 0$. Hence $\varphi = F^{1/s}$ is well defined, nonnegative and positively homogeneous of degree one, and its sublevel set $K \equiv \{z : \varphi(z) \leq 1\} = \{z : F(z) \leq 1\}$ is convex. By homogeneity, $z \in \kappa K$ if and only if $\varphi(z) \leq \kappa$, for every $\kappa > 0$.

    Let $x$ and $y$ be arbitrary, $\theta \in [0, 1]$, $a > \varphi(x)$ and $c > \varphi(y)$, so that $x / a$ and $y / c$ lie in $K$. With $e \equiv \theta a + (1 - \theta) c$, convexity of $K$ gives
    \begin{equation*}
        \theta x + (1 - \theta) y = e \left( \frac{\theta a}{e} \frac{x}{a} + \frac{(1 - \theta) c}{e} \frac{y}{c} \right) \in e K ,
    \end{equation*}
    so $\varphi(\theta x + (1 - \theta) y) \leq \theta a + (1 - \theta) c$. Letting $a \downarrow \varphi(x)$ and $c \downarrow \varphi(y)$ gives the convexity of $\varphi$.
\end{proof}

\begin{proof}[Proof of Lemma \ref{lem:finite}]
    The segment from $\bar{b}$ to $b$ lies in $\operatorname{conv} \mathcal{B} \subset U$. By condition \ref{def:regular:semialgebraic} of Assumption \ref{assumption:branch}, the curve $t \mapsto \zeta(t)$ is therefore semialgebraic on $[0, 1]$, and, being locally Lipschitz on a compact set, it is Lipschitz with some constant $\kappa_\zeta$.

    \textit{Finitely many events.} The graph of $\zeta$ over $(0, 1)$ is a finite disjoint union of Nash submanifolds \cite[Prop.~2.9.10]{bochnak1998}, each of them a point or a curve. On each curve, the projection to the $t$-axis has finitely many critical values \cite[Thm.~9.6.2]{bochnak1998}. Exclude these critical values, the $t$-coordinates of the points, and the endpoints of the projections of the curves, finitely many times in all. At any other time $t$, the graph meets the line over $t$ only at $(t, \zeta(t))$, so exactly one curve contains the graph near that point, and its projection to the $t$-axis is a local Nash diffeomorphism there. Hence $\zeta$ is Nash, and so analytic, near every time not excluded. Further, $F(\zeta(t))$ is semialgebraic in $t$, so its zero set is a finite union of points and intervals, and $\delta = F(\zeta)^{1/s}$ has finitely many isolated zeros.

    \textit{Lipschitz continuity.} By Lemma \ref{lem:phi_convex} and property \ref{def:deficiency:homogeneous} of Definition \ref{def:deficiency}, $\varphi$ is finite, convex and positively homogeneous of degree one, hence subadditive: $\varphi(x + y) = 2 \varphi(\frac{x + y}{2}) \leq \varphi(x) + \varphi(y)$. Thus $\varphi(x) - \varphi(y) \leq \varphi(x - y) \leq \kappa_\varphi \lVert x - y \rVert$, where $\kappa_\varphi$, the maximum of $\varphi$ over the unit sphere, is finite because a finite convex function is continuous. Exchanging $x$ and $y$, $\varphi$ is Lipschitz with constant $\kappa_\varphi$, and $\delta = \varphi \circ \zeta$ is Lipschitz with constant $\kappa_\varphi \kappa_\zeta$.

    \textit{One-sided derivatives.} Let $(\tau_0, \tau_1)$ be an open interval between consecutive events, or between an event and an endpoint of $[0, 1]$. There $F(\zeta(t))$ is analytic, so it either has no zeros or vanishes identically: an isolated zero would be an event, and a non-isolated zero of an analytic function on an interval forces it to vanish throughout. Hence $\delta$ is analytic on $(\tau_0, \tau_1)$, with $\delta'$ semialgebraic and $\lvert \delta' \rvert \leq \kappa_\varphi \kappa_\zeta$. If $\delta'(t)$ had no limit as $t \downarrow \tau_0$, a level $\kappa$ strictly between its lower and upper limits would make $\{t \in (\tau_0, \tau_1) : \delta'(t) > \kappa\}$ a union of infinitely many intervals accumulating at $\tau_0$, contradicting the first fact above. So $\delta'(t)$ converges as $t \downarrow \tau_0$, to a finite limit by the bound on $\delta'$. By the mean value theorem, the difference quotient $(\delta(\tau_0 + h) - \delta(\tau_0)) / h$ equals $\delta'$ at some point of $(\tau_0, \tau_0 + h)$, so the right derivative $\delta'(\tau_0^+)$ exists and equals this limit. Left derivatives are treated symmetrically, and the same argument, applied to each component of the Lipschitz, semialgebraic curve $\zeta$, shows that $\zeta$ also has finite one-sided derivatives at every time.
\end{proof}

Proposition \ref{prop:accounting} rests on the integrability of the two integrands it separates, which Lemma \ref{lem:integrable} establishes.

\begin{lemma}[Integrability] \label{lem:integrable}
    Under Assumption \ref{assumption:branch}, on every stretch on which $\delta > 0$, between consecutive events or between an event and an endpoint of $[0, 1]$, the functions $\delta''$, $\dot{\zeta}^\top \nabla^2\varphi(\zeta)\, \dot{\zeta}$ and $\nabla\varphi(\zeta)^\top \ddot{\zeta}$ are integrable, and
    \begin{equation} \label{eq:chain}
        \delta'' = \dot{\zeta}^\top \nabla^2\varphi(\zeta)\, \dot{\zeta} + \nabla\varphi(\zeta)^\top \ddot{\zeta} .
    \end{equation}
\end{lemma}
\begin{proof}
    Let $(\tau_0, \tau_1)$ be such a stretch, and let $\kappa_\zeta$ and $\kappa_\varphi$ be as in the proof of Lemma \ref{lem:finite}. Since $F$ is a polynomial and $F(\zeta(t)) > 0$ there, $\varphi = F^{1/s}$ is smooth in a neighborhood of each $\zeta(t)$, and since $\varphi$ does not depend on $t$, the chain rule gives \eqref{eq:chain}.

    The derivative $\delta'$ is bounded by $\kappa_\varphi \kappa_\zeta$ and continuously differentiable, with $\delta''$ semialgebraic, so $\delta''$ is integrable by the second fact above. Let $p(t) \equiv \nabla\varphi(\zeta(t))$. Each component $p_k$ is semialgebraic and continuously differentiable, and bounded by $\kappa_\varphi$, since the gradient of a Lipschitz function is bounded by its Lipschitz constant. By the second fact, each $\dot{p}_k$ is integrable. Since $\dot{p} = \nabla^2\varphi(\zeta)\, \dot{\zeta}$, the first term on the right of \eqref{eq:chain} equals $\dot{p}^\top \dot{\zeta}$, which is bounded in absolute value by $\kappa_\zeta \sum_k \lvert \dot{p}_k \rvert$ and is therefore integrable. The second term, the difference of two integrable functions, is integrable as well.
\end{proof}

\begin{proof}[Proof of Proposition \ref{prop:accounting}]
    Let $0 = t_0 < t_1 < \dots < t_K < t_{K+1} = 1$, with $t_1, \dots, t_K$ the events. By Lemma \ref{lem:finite} and its proof, on each stretch $(t_k, t_{k+1})$ the function $\delta$ is analytic and either positive throughout or identically zero, and $\delta'$ has the finite one-sided limits $\delta'(t_k^+)$ and $\delta'(t_{k+1}^-)$. On a stretch where $\delta > 0$, $\delta''$ is integrable by Lemma \ref{lem:integrable}, so $\delta'$, extended by its one-sided limits, is absolutely continuous on the closed stretch; on a stretch where $\delta \equiv 0$, $\delta' \equiv 0$. On either kind of stretch, integration by parts and $\int_{t_k}^{t_{k+1}} \delta' = \delta(t_{k+1}) - \delta(t_k)$, valid for the Lipschitz $\delta$, give
    \begin{align*}
        & \int_{t_k}^{t_{k+1}} (1 - \xi)\, \delta''(\xi)\, \mathrm{d}\xi \\
        & \quad = (1 - t_{k+1})\, \delta'(t_{k+1}^-) - (1 - t_k)\, \delta'(t_k^+) \\
        & \qquad + \delta(t_{k+1}) - \delta(t_k) .
    \end{align*}
    Summing over $k$, the first two terms telescope to $-\delta'(0^+) - \sum_{i=1}^K (1 - t_i) J_i$, so
    \begin{equation*}
        \delta(1) = \delta(0) + \delta'(0^+) + \int_0^1 (1 - \xi)\, \delta''(\xi)\, \mathrm{d}\xi + \Delta_{\mathrm{A}} .
    \end{equation*}
    The deficiency $\delta = v^{1/s}$ is differentiable at $\bar{b}$, where $v(\bar{b}) > 0$ (Section \ref{sec:branches}), so $\nabla\delta(\bar{b})^\top d = \delta'(0^+)$ and $\Delta = \delta(0) + \delta'(0^+) - \delta(1)$. The set $\Omega$ is the union of the stretches on which $\delta > 0$. On $\Omega$, \eqref{eq:chain} splits $\delta''$ into the integrands of $\Delta_{\mathrm{N}}$ and $\Delta_{\mathrm{R}}$, both integrable by Lemma \ref{lem:integrable}; on the other stretches, $\delta'' = 0$, and the events are finitely many. Hence $\int_0^1 (1 - \xi)\, \delta''(\xi)\, \mathrm{d}\xi = \Delta_{\mathrm{N}} + \Delta_{\mathrm{R}}$, with both terms finite, and $\Delta = -\Delta_{\mathrm{A}} - \Delta_{\mathrm{N}} - \Delta_{\mathrm{R}}$.
\end{proof}

\begin{proof}[Proof of Proposition \ref{prop:atomic}]
    Fix $t \in (0, 1)$. The set $\Lambda(t)$ is nonempty and compact by Lemma \ref{lem:directional} and its proof, and convex as the solution set of linear equations and inequalities in $(\lambda, \mu)$. Lemma \ref{lem:directional}, applied at $\bar{b} + td$ in the directions $d$ and $-d$, gives $v'(t^+) = \max_{\Lambda(t)} \partial_t L$ and $v'(t^-) = -\max_{\Lambda(t)} (-\partial_t L) = \min_{\Lambda(t)} \partial_t L$.

    Since $\zeta(t_i)$ is feasible, $v(t_i) \geq v^*(t_i) \geq 0$ by property \ref{def:deficiency:contract} of Definition \ref{def:deficiency}, so either $v(t_i) > 0$ or $\delta(t_i) = 0$. Where $v(t_i) > 0$, the chain rule for one-sided derivatives gives $\delta'(t_i^\pm) = \frac{1}{s} v(t_i)^{1/s - 1} v'(t_i^\pm)$, which is \eqref{eq:atom_interior}, nonnegative as a maximum less a minimum over the same set. Where $\delta(t_i) = 0$, the one-sided derivatives exist by Lemma \ref{lem:finite}, and $\delta \geq 0$ forces $\delta'(t_i^-) \leq 0 \leq \delta'(t_i^+)$, which gives \eqref{eq:atom_boundary} and $J_i \geq 0$.
\end{proof}

\begin{proof}[Proof of Theorem \ref{thm:validity}]
    Each $J_i$ is nonnegative by Proposition \ref{prop:atomic} and each weight $1 - t_i$ is positive, so $\Delta_{\mathrm{A}} \geq 0$. On $\Omega$, $\varphi$ is convex by Lemma \ref{lem:phi_convex} and smooth near the branch, so its Hessian is positive semidefinite there and the integrand of $\Delta_{\mathrm{N}}$ is nonnegative; hence $\Delta_{\mathrm{N}} \geq 0$. Since all three terms of \eqref{eq:accounting} are finite by Proposition \ref{prop:accounting}, \eqref{eq:overestimation} follows.
\end{proof}
\fi

\section*{Acknowledgment}
The authors thank Alan Valenzuela for his assistance with the experiments. The authors used Claude (Anthropic) to help draft and edit the text throughout and to write the code for the figures and tables.

\bibliography{ref.bib}

\vfill

\end{document}